\documentclass[12pt, reqno, fleqn]{article}
\usepackage[margin=1in]{geometry}
\usepackage{lmodern}
\usepackage[T1]{fontenc} 
\usepackage{natbib}
\usepackage[figuresright]{rotating}
\usepackage{float}
\usepackage{graphicx}
\usepackage{epstopdf}
\usepackage{url}
\usepackage[colorlinks,citecolor=blue,linkcolor=blue,urlcolor=red,pagebackref]{hyperref}
\usepackage{amsmath,amsthm}
\usepackage{longtable}
\usepackage[dvipsnames]{xcolor}
\usepackage[classfont = sanserif, langfont = roman, funcfont = italic]{complexity}
\usepackage{amssymb}
\usepackage[yyyymmdd,hhmmss]{datetime}
\usepackage{algpseudocode}
\usepackage{algorithm}
\usepackage{orcidlink}

\newtheorem{theorem}{Theorem}
\newtheorem{lemma}{Lemma}
\newtheorem{conjecture}{Conjecture}
\newtheorem{corollary}{Corollary}

\theoremstyle{definition}

\theoremstyle{remark}
\newtheorem{definition}{Definition}

\providecommand{\keywords}[1]
{
  \small	
  \textbf{Keywords:} #1
}

\newclass{\SHARPP}{\#P}
\newclass{\PPOLY}{P/Poly}

\newlang{\OV}{OV}
\newlang{\HAMPATH}{HAMPATH}
\newlang{\HAMCYCLE}{HAMCYCLE}
\newlang{\CLIQUE}{CLIQUE}
\newlang{\MULT}{MULT}
\newlang{\DLP}{DLP}

\newfunc{\LDE}{LDE}
\newfunc{\PER}{PER}

\begin{document}

\title {New Techniques for Constructing Rare-Case Hard Functions}                      

\author{Tejas Nareddy\orcidlink{0009-0007-7032-6654}\footnote{Department of Computer Science and Information Systems, Birla Institute of Technology and Science, Pilani, Pilani-333031, Rajasthan, I\textsc{ndia}. Email: \texttt{f20211462@pilani.bits-pilani.ac.in}.} \and 
Abhishek Mishra\orcidlink{0000-0002-2205-0514}\footnote{Department of Computer Science and Information Systems, Birla Institute of Technology and Science, Pilani, Pilani-333031, Rajasthan, I\textsc{ndia}. Email: \texttt{abhishek.mishra@pilani.bits-pilani.ac.in}.}}

\maketitle
\thispagestyle{empty}

\begin{abstract}
We say that a function is rare-case hard against a given class of algorithms (the adversary) if all algorithms in the class can compute the function only on an $o(1)$-fraction of instances of size $n$ for large enough $n$. Starting from any $\NP$-complete language, for each $\alpha > 0$, we construct a function that cannot be computed correctly even on a $1/n^\alpha$-fraction of instances for polynomial-sized circuit families if $\NP \not \subset \PPOLY$ and by polynomial-time algorithms if $\NP \not \subset \BPP$ - functions that are rare-case hard against polynomial-sized circuits and polynomial-time randomized algorithms. The constructed function is a number-theoretic polynomial evaluated over specific finite fields. For $\NP$-complete languages that admit parsimonious reductions from all of $\NP$ (for example, $\SAT$), the constructed functions are hard to compute even on a $1/n^\alpha$-fraction of instances by polynomial-time randomized algorithms and polynomial-sized circuit families simply if $\P^{\SHARPP} \not \subset \BPP$ and $\P^{\SHARPP} \not \subset \PPOLY$, respectively. We also show that if the \textit{Randomized Exponential Time Hypothesis (RETH)} is true, none of these constructed functions can be computed even on a $1/n^\alpha$-fraction of instances in subexponential time. These functions are very hard, almost always.

While one may not be able to efficiently compute the values of these constructed functions themselves, in polynomial time, one can verify that the evaluation of a function, $s = f(x)$, is correct simply by asking a prover to compute $f(y)$ on targeted queries.

We have extended our work to give an alternative proof of a variant of Lipton's theorem \citep{Lipton1989}. We also compare our techniques for constructing rare-case hard functions with two other existing methods in the literature \citep{Sudan2001, Feige1996}. Our novel contributions can be summarized as follows:

\begin{enumerate}

\item A careful construction of the generalized certificate counting polynomials that admit a downward self-reduction and a sumcheck protocol. This construction may also find applications in other areas of computational complexity.

\item Application of the Chinese remainder theorem to reduce the field size.

\item Application of a theorem on an upper bound on the gap between consecutive primes \citep{Baker2001} to prove rare-case hardness of our number-theoretic polynomial.

\end{enumerate}

\end{abstract}


\newpage

\pagenumbering{roman}
\setcounter{page}{1}

\maketitle

\tableofcontents

\newpage

\pagenumbering{arabic}
\setcounter{page}{1}

\section{Introduction}
\label{section:1}

For decades, complexity theory has focused chiefly on worst-case hardness, from the original proofs of \cite{Cook1971} and \cite{Levin1973} that the \textit{satisfiability} language ($\SAT$) is $\NP$-complete to \cite{Karp1972} showing the following year that many natural languages are $\NP$-complete as well. These languages are not solvable by deterministic polynomial-time algorithms if $\P \neq \NP$. However, for many applications, cryptography being the foremost, we want better guarantees of hardness than just worst-case hardness. It is not enough for our cryptographic protocols, that for every algorithm, there is some instance that is hard. This motivates the need for ``rare-case'' hardness. Suppose we can guarantee that for some problem, for any reasonably fast algorithm, the algorithm only outputs the correct answer on an $o(1)$-fraction of instances. In that case, we can be assured that, for large enough $n$, any instance we randomly generate will probably not be solvable by a reasonably fast adversary.

The phrase ``rare-case hardness'' is inspired by its usage by \cite{Goldreich2018} on counting $t$-cliques, where they show that counting cliques of a specific size in a graph is hard for even a $1/\polylog(n)$-fraction of instances if it is in the worst case. Similar work has been done to show that some variants of $t$-clique are as hard in the average-case as they are in the worst case \citep{Mina2020, Enric2019}. Similar results have been shown by \cite{Kane2019} for the orthogonal vectors ($\OV$) problem against $\AC^0$ formulas under certain worst-case hardness assumptions. They have shown the existence of a distributional $\OV$ problem that can be solved by $o(n^2)$-sized $\AC^0$ circuits for a $1 - o(1)$-fraction of instances.

As a motivational example, consider the problem of multiplying two $n$-bit numbers ($\MULT_n$). \cite{Harvey2021} have proved that $\MULT_n$ can be solved in $O (n \log n)$-time on a \textit{multitape Turing machine (MTM)}. We can say that $\MULT_n$ is easy for the set of $O(n \log n)$-time MTMs, since there exists at least one MTM that solves $\MULT_n$ correctly over all instances with parameter $n$. It is an open problem whether there exists an $O(n)$-time MTM which correctly solves $\MULT_n$ on all instances with parameter $n$ \citep{Afshani19}. Now we ask the question: what is the largest fraction of instances an $O(n)$-time MTM can solve $\MULT_n$? If the answer to this question is $1$, then we say that $\MULT_n$ is easy for the set of $O(n)$-time MTMs. If the answer is a constant, we say that $\MULT_n$ is average-case hard  \citep{Ball2017} for the set of $O(n)$-time MTMs. Finally, if the answer is a negligible fraction that tends to $0$ as $n$ tends to infinity, we say that $\MULT_n$ is rare-case hard (formally defined in Section \ref{section:3}) for the set of $O(n)$-time MTMs.

Another famous and instructive example of rare-case hardness is the usage of the \textit{Discrete Logarithm Problem} ($\DLP$) in the pseudorandom generator of \cite{Blum1982}, depending on the worst-case hardness of the $\DLP$. The $\DLP$ asks whether given a prime $p$ of $n$-bits, a multiplicative generator $g$ of $\mathbb{Z}^*_p$, and $l \in \mathbb{Z}^*_p$, to find $r$ such that $g^r \equiv l \mod p$. Suppose for any $\alpha > 0$, we have a polynomial-time algorithm (an oracle $O$) solving the $\DLP$ on a $1/n^\alpha$-fraction of instances for $n$-bit primes. We have a simple worst-case to rare-case reduction (formally defined in Section \ref{section:3}) - given $l$, simply generate $r^\prime$ at random and ask $O$ for the answer to the $\DLP$ for $l \cdot g^{r^\prime}$. If $O$ returns $r$, check if $g^{r-r^\prime} \equiv l \mod p$, and return $r - r^\prime$ if so. Otherwise, we will repeat this process. We are expected to find the answer in $n^\alpha$ queries, giving us a probabilistic algorithm. Due to this, we have that if the $\DLP$ is not solvable by randomized polynomial-time algorithms, then no randomized or deterministic algorithm solves the $\DLP$ on a $1/n^\alpha$-fraction of instances for any $\alpha > 0$, giving us a one-way function.

However, we would also like to construct families of hard problems that are hard due to many weak conjectures and hypotheses, which, when scaled down to asymptotically small input sizes, can also give us protocols such as proof of work \citep{Dwork1993}, that are hard to solve almost all the time, but always very quick to verify\footnote{Say, under some conjecture, taking $\Omega(n^2)$-time for a prover to solve, but $\polylog(n)$-time to verify.}.

In this paper, we show that we can construct infinite families of such rare-case hard functions using $\NP$-complete languages as our starting point. The constructed functions are number-theoretic polynomials evaluated over $\mathbb{Z}_p$ for certain primes $p$. These families also have polynomial-time interactive proof systems where the prover and verifier only need to communicate inputs and outputs to the constructed function for a verifier to be convinced. In fact, the interactive proof system is used within the reduction. Interestingly, we can look at any reduction as an interactive proof with varying degrees of trust. Many-one polynomial-time reductions for $\NP$-completeness fully trust the prover and take the prover's word as gospel. Here, since our hypothetical algorithm is correct only sometimes, we do not trust it fully but employ some tools to help extract more truth from an oracle that tells the truth only sometimes. A notable work that uses a verification protocol as a reduction is by \cite{Shamir1992} that proves $\IP = \PSPACE$. We use a modified version of the sumcheck protocol as proposed by \cite{Lund1992}.

We use a theorem of \cite{Sudan2001} that \cite{Goldreich2018} use to error-correct to go from an algorithm that is correct on a small fraction of instances to a randomized algorithm that is correct with very high probability on any instance. As with this paper, and most other works on average-case hardness \citep{Ball2017}, we leverage the algebraic properties of low-degree polynomials over large fields to show that if such a polynomial is ``sufficiently expressive,'' in that it can solve a problem we believe to be hard in the worst-case with a small number of evaluations of the polynomial, we can error-correct upwards from the low-correctness regime to solve our problem that is conjectured to be hard with a very high probability.

The remainder of the paper is organized as follows. Section \ref{section:2} gives the preliminaries. Section \ref{section:3} gives an overview of our results. Section \ref{section:4} describes the generalized certificate counting polynomials. Section \ref{section:5} gives an oracle sumcheck protocol over $\mathbb{Z}_p$. Section \ref{section:6} gives a method to reconstruct the certificate counting polynomials over $\mathbb{Z}$. Section \ref{section:7} proves the main results of the paper. Section \ref{section:8} compares our techniques to previously known techniques of constructing rare-case hard functions. Finally, we give further directions for research in Section \ref{section:9}. The appendix gives an alternative proof of a variant of Lipton's theorem \citep{Lipton1989}.

\section{Preliminaries}
\label{section:2}

In this section, we briefly explore the ideas that are used in the proofs and reductions. Some subsections will elaborate slightly more than necessary to impart ``intuitive pictures'' or ideas to keep in mind that will help one better digest the proofs and the larger ideas that motivate the proofs. The lemmas and theorems are specialized to our requirements and are presented as lemmas.

\subsection{Notations}
\label{section:2.1}

$\mathbb{N}$ denotes the set of natural numbers, $\{\, 1, 2, 3, 4, \ldots \,\}$. For all $n \in \mathbb{N}$, $[n]$ denotes the set of first $n$ natural numbers, $\{\, 1, 2, \ldots, n - 1, n \,\}$. $\mathbb{Z}$ denotes the ring of integers with the usual addition and multiplication operations. The variable $p$ denotes a prime number. $Z_p$ is the finite field with the usual operations of addition and multiplication modulo $p$. $\mathbb{Z}^*_p$ is the finite multiplicative group with the group operation as multiplication modulo $p$. $\mathbb{F}$ denotes a finite field. The notation $(a, b)_p$ denotes the set of primes in the interval $(a, b)$. The notation $\pi(a, b)$ denotes the number of primes in the interval $(a, b)$.

$O$ denotes an oracle for computing some function. $M^O$ denotes that the machine $M$ has oracle access to $O$. The function $\poly(n)$ denotes any polynomial in $n$. The function $\polylog(n)$ denotes any polynomial in $\log n$. The function $\ln x$ is the natural logarithm of $x$ to the base $e$. $\mathcal{P}[E]$ denotes the probability of the event $E$. $\mathcal{E}[X]$ denotes the expectation of the random variable $X$. The notation $f(x) \sim g(x)$ means that $\lim_{x \to \infty} f(x) / g(x) = 1.$

The notation $(x_i)_{i = 1}^n$ denotes the ordered $n$-tuple, $(x_1, x_2, \ldots, x_{n - 1}, x_n)$. For simplifying the notations, we use the comma operator, ``,'', between two $n$-tuples to mean the ``Cartesian product'' of the two $n$-tuples (the ``$n$'' can be different for the two operands). For example, 
\begin{equation*}
\begin{split}
\left( (x_i)_{i = 1}^3, (x_i)_{i = 5}^9 \right) & = \left( (x_i)_{i = 1}^3 \times (x_i)_{i = 5}^9 \right) \\
& = (x_1, x_2, x_3, x_5, x_6, x_7, x_8, x_9).
\end{split}
\end{equation*}

For a language $L \subset \{\, 0, 1 \,\}^*$, we define its characteristic function $L : \{\, 0, 1 \,\}^* \to \{\, 0, 1 \,\}$ as
\begin{equation*}
L(y) = 1 \iff y \in L,
\end{equation*}
and
\begin{equation*}
L(y) = 0 \iff y \notin L.
\end{equation*}

\subsection{The Schwartz-Zippel Lemma}
\label{section:2.2}

One of the key factors even allowing the existence of many modern error-correcting codes \citep{Reed1960, Gemmell1992} is the fact that polynomials whose degree is much smaller than the size of the field it is evaluated on are very rarely $0$. More concretely, analogous to the fundamental theorem of algebra over the complex plane, the Schwartz-Zippel lemma for finite fields says that any multivariate polynomial $f : \mathbb{F}^n \to \mathbb{F}$ of degree $d$ can take the value $0$ on at most a $d / |\mathbb{F}|$-fraction of inputs. That is,
\begin{lemma} 
\label{lemma:1}
\textbf{The Schwartz-Zippel Lemma \citep{Schwartz1980, Zippel1979}.} \\
If $x$ is randomly chosen from $\mathbb{F}^n$, then  
\begin{equation*}
\mathcal{P}_{x \leftarrow_r \mathbb{F}^n} [f(x) = 0] \leq \frac{d}{|\mathbb{F}|}.
\end{equation*}
Even more generally, for $S \subseteq \mathbb{F}$,
\begin{equation*}
\mathcal{P}_{x \leftarrow_r S^n} [f(x) = 0] \leq \frac{d}{|S|}.
\end{equation*}
\end{lemma}
Due to this lemma, we can also see that low-degree polynomials cannot take any one value in $\mathbb{F}$ too often and that two low-degree polynomials cannot agree too often. This enables the existence of error-correcting codes and list-decoders \citep{Sudan2001}.

\subsection{The List Decoding of Polynomials Problem}
\label{section:2.3}

We say that a function $g$ $\epsilon$-agrees ($0 \leq \epsilon \leq 1$) with a function $f$, if the two functions return the same values on $\epsilon$-fraction of the inputs. Let $l(\epsilon, d)$ be the number of the polynomials $g$ with a total degree at most $d$ and having $\epsilon$-agreement with $f$. In the list decoding of polynomials problem, we are given an oracle $O$ for computing a function $f: \mathbb{F}^n \to \mathbb{F}$. We are also given the parameters $\epsilon \in [0, 1]$ and $d \in \mathbb{N}$. Our objective is to construct randomized oracle machines $\left( M_i^O \right)_{i = 1}^{l(\epsilon, d)}$ such that for every polynomial $g$ of total degree at most $d$ and having $\epsilon$-agreement with $f$, there exists a randomized oracle machine $M_i^O$ ($i \in [l(\epsilon, d)]$) computing $g$ with a probability of error upper bounded by $1/2^{q(n)}$, where $q(n)$ is a polynomial. The list-decoder we will be using throughout this work is due to the following theorem:

\begin{lemma} 
\label{lemma:2}
\textbf{The Sudan-Trevisan-Vadhan (STV) List-Decoder \citep{Sudan2001}.} \\
Given any oracle $O$ that computes a polynomial $p : \mathbb{F}^n \to \mathbb{F}$ of degree $d$ correctly on over an $\epsilon > \sqrt{2d / |\mathbb{F}|}$ fraction of instances, in $\poly(n, d, 1 / \epsilon, \log |\mathbb{F}|)$-time, we can produce $O(1 / \epsilon)$ randomized oracle machines (with oracle access to $O$), all of which compute some multivariate polynomial from $\mathbb{F}^n$ to $\mathbb{F}$ of degree $d$, one of which computes $f$. Moreover, each machine runs in $\poly(n, d, 1 / \epsilon, \log |\mathbb{F}|)$-time and disagrees with the polynomial it intends to compute with a probability of at most $1 / 2^{q(n)}$ for some polynomial $q$.
\end{lemma}

The list-decoder works by trying to compute all polynomials with an $\Omega(\epsilon)$-fraction agreement with $O$ and then taking random lines in $\mathbb{F}^n$ parameterized by one variable in the univariate case to reconstruct these polynomials. We will, however, use this result as a black box in all our proofs. We will call this the ``STV list-decoder'' going forward.

As a prelude to future sections, we aim to error-correct from $1 / n^\alpha$-correctness. Notice that when $1 / \epsilon$, $d$ and $|\mathbb{F}|$ are polynomials in $n$, the entire procedure runs in $\poly(n)$-time. Once we have $O(1 / \epsilon) = \poly(n)$ machines, we employ various techniques to ``identify'' which machine computes the polynomial $f$ that interests us.

\subsection{The Chinese Remainder Theorem}
\label{section:2.4}
An age-old theorem we will use from elementary number theory is the Chinese remainder theorem \citep{Niven1991}. It gives a polynomial-time algorithm for solving a given set of linear congruences.
\begin{lemma} 
\label{lemma:3}
\textbf{The Chinese Remainder Theorem.} \\
For a given set of distinct primes $(p_i)_{i = 1}^n$ and a set of integers $(a_i)_{i = 1}^n$, such that $a_i \in Z_{p_i}$ for all $i$, the system of linear congruences $x \equiv a_i \mod p_i$ has a unique solution modulo $\prod_{i = 1}^n p_i$, that can be computed in polynomial-time in the input size.
\end{lemma}

Specifically, we compute the number of accepting certificates modulo $p$ for many primes $p$ and find the number of certificates by ``Chinese remaindering''. As long as the product of the primes is larger than the largest number of accepting certificates, we are guaranteed to get our solution.

\subsection{The Distribution of Primes}
\label{section:2.5}
One thing that we want to be sure of is that we have enough primes that are ``of similar size''. This is important for us because if our oracle $O$ is correct on a large fraction of instances for some very large prime, we may satisfy the case where $O$ is correct on the required number of instances over all primes just by being sufficiently correct over the field of one very large prime. To avoid this, we would like to ensure that there are many primes of roughly similar size, ensuring that sufficiently many primes have sufficient correctness. This will be proved in later sections. In this section, we will present the lemmas, theorems, and ideas.

A landmark theorem describing the distribution of the primes is the prime number theorem, proven independently by \cite{Hadamard1896} and \cite{Poussin1896}. It states that if $\pi(x)$ is the number of primes less than $x$, then
\begin{equation*}
  \pi(x) \sim \frac{x}{\ln x}.
\end{equation*}
This theorem alone is not good enough for us. The conjecture of \cite{Cramer1936} states that $p_{n+1} - p_n = O \left( (\log p_n)^2 \right)$\footnote{$p_n$ refers to the $n$'th prime number.}  and this would suffice for us. However, this problem is open and is stronger than the upper bounds implied by the Riemann hypothesis \citep{Riemann1859}. The following theorem is the strongest known unconditional upper bound on gaps between consecutive primes.
\begin{lemma} 
\label{lemma:4}
\textbf{An Upper Bound on the Gap Between Consecutive Primes \citep{Baker2001}.} \\
An upper bound on the difference between consecutive primes, $p_{n + 1}$ and $p_n$, for all $n \in \mathbb{N}$ is given by
\begin{equation*}
  p_{n+1} - p_n = O \left( p_n^{0.525} \right).
\end{equation*}
\end{lemma}
This is good enough for our purposes. In particular, the gap between consecutive primes between $m$ and $2m$ is at most $m^{0.526}$ for sufficiently large $m$, giving us at least $m^{0.474}$ primes in this range. The result itself uses deep techniques in analytic number theory and sieve theory, and the interested reader is directed to the original paper.

\subsection{The Sumcheck Protocol}
\label{section:2.6}
The technique of \cite{Lund1992} to verify answers to polynomial queries set the field of interactive proofs ablaze, famously followed by a proof by \cite{Shamir1992} that $\IP = \PSPACE$. The protocol is described below.

Suppose we have a polynomial $g:\mathbb{F}^n \to \mathbb{F}$ of degree $d$ with $dn < |\mathbb{F}|$. The sumcheck protocol begins with the prover making the following claim:
\begin{equation*}
  s = \sum \limits_{x \in \{\, 0,1 \,\}^n} g(x).
\end{equation*}
Along with this, the prover also sends the verifier the coefficients of the univariate polynomial,
\begin{equation*}
  g^\prime(r) = \sum \limits_{(x_i)_{i = 2}^{n} \in \{\, 0,1 \,\}^{n-1}} g \left( r, (x_i)_{i = 2}^{n} \right).
\end{equation*}
The verifier checks that the degree of $g^\prime(r)$ is at most the degree of $x_1$ in $g$ and that $g\prime(0) + g\prime(1) = s$. If true, the verifier picks a random $r^\prime$ from $\mathbb{F}$ and iterates the process, asking the prover to prove that $g^\prime(r^\prime)$ is indeed the value computed by the verifier. In the last step, in the $n$'th iteration, the verifier has to evaluate the polynomial $g$ on some input in $\mathbb{F}^n$. If this evaluation is as suggested by the execution of the protocol, then the verifier accepts. If, at any stage, the verifier receives a polynomial whose degree is too large or whose evaluation is inconsistent, it rejects.

Note that if the claim is correct, the prover can remain entirely truthful and give all answers truthfully - the verifier accepts with probability $1$. If the claim is wrong, due to the Schwartz-Zippel Lemma (\ref{lemma:1}), the probability that $g^\prime(r^\prime)$ is the same as the correct summation of $g$ in any step is at most $d/|\mathbb{F}|$. By induction, one can show that the probability that the verifier accepts if the initial claim is incorrect is at most $dn/|\mathbb{F}|$. It is key to remember that even if the prover lies cleverly, managing to pass iterations $1$ through $n-1$, with high probability, it will be exposed in the last step, depending on the random value in $\mathbb{F}$ chosen in the last step by the verifier.

Using the list decoding techniques of the STV list-decoder (Lemma \ref{lemma:2}), we can ask questions to an oracle $O$ that knows the answer sometimes and sometimes answers questions different from the ones we ask it. Using the sumcheck protocol, we can find out answers to these questions even when $O$'s answers are very ``noisy''.

\section{An Overview of our Results}
\label{section:3}

Here, our main intention is to show that there is a reduction from $\NP$-complete languages to a particular ``set'' of polynomial evaluations. To be more concrete, suppose we have an oracle $O$ that computes the number of certificates for any instance $x$ of some $\NP$-complete language modulo $p$ for sufficiently many primes $p$. We can use the Chinese remainder theorem (Lemma \ref{lemma:3}) to reconstruct the exact number of certificates certifying $x$. Moreover, one might notice that if the $\NP$-complete language has parsimonious reductions from all of $\NP$ (for example, $\SAT$ \citep{Arora2009, Goldreich2008}), then this ``set'' of polynomial evaluations can, due to the parsimonious polynomial-time reductions, compute the number of accepting certificates for any language in $\NP$.

Before going into further details, the main idea of the reduction is that all reductions are interactive proofs between a main machine and an oracle, especially when the oracle is wrong some portion of the time. Our idea is to query an oracle that is correct on a $1 / n^\alpha$-fraction of instances, and along with the STV list-decoder (Lemma \ref{lemma:2}), we make a ``best effort'' attempt to find the answer to a query we have. In some cases, even with the list decoder, we will not be able to recover answers reliably, which means that any answers we receive must go through some scrutiny. We want to make sure any answers we use in further computations are sound with very high probability. We use the sumcheck protocol (Section \ref{section:2.6}), inspired by \cite{Shamir1992} to prove $\IP = \PSPACE$.

Before proceeding further, we formally define rare-case hardness and worst-case to rare-case reductions.

\begin{definition}
\label{def:1}
\textit{\textbf{Rare-Case Hard Functions.}} \\
\textit{Let $\mathcal{C}$ be a class of algorithms, circuits, decision trees, or any objects (or machines) of some model of computation. We say that a function $f$ is easy against $\mathcal{C}$ if there exists a machine in $\mathcal{C}$ that correctly computes $f$ on all the instances. We say that $f$ is hard against $\mathcal{C}$ if none of the machines in $\mathcal{C}$ can correctly compute $f$ on all the instances. We say that $f$ is $h(n)$-hard against $\mathcal{C}$ if no machine in $\mathcal{C}$ can compute $f$ correctly on an $\Omega(h(n))$-fraction\footnote{$h$ is a function defined from $\mathbb{N}$ to $\mathbb{R}$.} of  instances of length $n$ for all sufficiently large $n$. We say that $f$ is average-case hard against $\mathcal{C}$ if $f$ is $h(n)$-hard against $\mathcal{C}$ and $h(n) = \Theta(1)$. We say that $f$ is rare-case hard against $\mathcal{C}$ if $f$ is $h(n)$-hard against $\mathcal{C}$ and $h(n) = o(1)$.}
\end{definition}
\begin{definition}
\label{def:2}
\textit{\textbf{Worst-Case to Rare-Case Reductions.}} \\
\textit{We say that there is a $(t(n), h(n))$-worst-case to rare-case reduction from a function $f$ to a function $f^\prime$ if there exists an $O(t(n))$-time probabilistic algorithm that, given access to an oracle $O$ that computes $f^\prime$ correctly on an $h(n)$-fraction of instances of size $n$, where $h(n) = o(1)$, computes $f$ correctly with error probability less than $1/3$.}
\end{definition}

These worst-case to rare-case reductions are particularly interesting when $f$ is not believed to have polynomial-time probabilistic algorithms and $t(n) = n^{O(1)}$ since they imply polynomial-time algorithms cannot compute $f^\prime$ even on a vanishingly small fraction of instances. If $f$ is not believed to have $2^{(1-\epsilon)k(n)}$-time algorithms, in the case where $t(n) = n^{O(1)}$, neither should $f^\prime$, but even for an $o(1)$-fraction of instances.

Mainly, this paper shows that from any $\NP$-complete language, $f: \{\, 0, 1 \,\}^* \to \{\, 0, 1 \,\}$, we can construct a function $f^\prime$, such that in polynomially many queries to an oracle $O$ that computes $f'$ and is almost always wrong, we can compute $f$ with very low error probability. In fact, for every $\alpha > 0$, we can construct $f^\prime$ such that $O$ computing $f^\prime$ correctly on a $1/n^\alpha$-fraction of instances is sufficient to compute $f$ in $n^{O(1)}$ queries to $O$. Notably, if $\NP \not\subset \PPOLY$, which is widely believed to be true due to the famous theorem of \cite{Karp1980}, then for any polynomial-sized family of circuits $\{\, C_n \,\}_{n \in \mathbb{N}}$,
\begin{equation*}
\mathcal{P}_{x \leftarrow_r \mathbb{F}^m}[f^\prime(x) = C_n(x)] < \frac{1}{n^\alpha},
\end{equation*}
for all sufficiently large $n$. There is also a protocol by which a verifier $V$ can verify in polynomial time whether a prover, $P$'s claim that $s = f^\prime(x)$ simply by asking $P$ to compute $f^\prime(y)$ for sequences of $y$ chosen by $V$. We will see later that this works even when $P$ can only give an answer on a vanishing but sufficient fraction of queries.

Similar hardness results can be shown from conjectures such as $\NP \not\subset \BPP$, and weaker hypotheses such as $\P^{\SHARPP} \not\subset \PPOLY$ and $\P^{\SHARPP} \not\subset \BPP$ for $\NP$-complete languages known to have parsimonious reductions from all of $\NP$. Under conjectures such as the \textit{Randomized Exponential Time Hypothesis (RETH)} and the \textit{Randomized Strong Exponential Time Hypothesis (RSETH)} \citep{Impagliazzo2001, Calabro2009, Impagliazzo2001b}, which we will state below, we either have very strong hardness results for $f^\prime$ or a path to refutation for these hypotheses, via a barely efficient algorithm on a vanishing fraction of instances for $f^\prime$.

\begin{conjecture}
\label{conjecture:1}
\textbf{Randomized Exponential Time Hypothesis \citep{Dell2014}.} \\
There is an $\epsilon > 0$ such that no probabilistic algorithm correctly decides $3\SAT$ on $n$ variables with correctness probability larger than $2/3$ in time $2^{\epsilon n}$.
\end{conjecture}
\begin{conjecture}
\label{conjecture:2}
\textbf{Randomized Strong Exponential Time Hypothesis \citep{Dell2014, Stephens2019}.} \\
For every $\epsilon > 0$, there is a $k \in \mathbb{N}$ such that no probabilistic algorithm decides $k\SAT$ correctly with correctness probability larger than $2/3$ in time $2^{(1-\epsilon)n}$.
\end{conjecture}

\section{Generalized Certificate Counting Polynomials}
\label{section:4}
To go forward, first, suppose we have a language $L$ that is $\NP$-complete. $L$ has a verifier $V_L$ that takes in a certificate $z$ of length $n^c$ when the length of the instance $x$ is $n$. Due to the theorems of \cite{Cook1971} and \cite{Levin1973}, from the algorithm of $V_L$, we can compute a $\poly(n)$-sized circuit $C_L$ that takes in $x$ and $z$ as input and outputs $1$ if $z$ certifies that $x \in L$ and $0$ otherwise. We show below in this section that we have constant depth verification circuits for every $\NP$ language, from which we can construct verification polynomials.

Suppose we have a circuit $C_L$ with $n + n^c$ bits of input. We allow the gates of $C_L$ to be of unbounded fan-in. We construct the following polynomial over $\mathbb{Z}_p$, by the following rules:

\begin{enumerate}

\item Let the input variables be $x = (x_i)_{i = 1}^n$ and $z = (z_j)_{j = 1}^{n^c}$. For each input $x_i$ or $\neg x_i$ ($i \in [n]$), the corresponding polynomials are $x_i$ and $1 - x_i$, respectively. Similarly, for each input $z_j$ or $\neg z_j$ ($j \in \left[ n^c \right]$), the corresponding polynomials are $z_j$ and $1 - z_j$, respectively. 

\item For each \texttt{AND} gate with $k$ inputs from gates whose corresponding polynomials are $(g_j)_{j = 1}^k$, the \texttt{AND} gate's corresponding polynomial is $\Pi_{j = 1}^k g_j$.

\item For each \texttt{OR} gate with $k$ inputs from gates whose corresponding polynomials are $(g_j)_{j = 1}^k$, the \texttt{OR} gate's corresponding polynomial is $1 - \Pi_{j = 1}^k (1 - g_j)$.

\item For each \texttt{NOT} gate with input from a gate with corresponding polynomial $g$, the corresponding polynomial for the \texttt{NOT} gate is $1 - g$.

\end{enumerate}

Let $g_{C_L, p}: \mathbb{Z}_p^{n + n^c} \to \mathbb{Z}_p$ be the polynomial corresponding to the output gate of $C_L$. Note that $g_{C_L, p}$ is a multivariate polynomial with coefficients in $\mathbb{Z}_p$ for which $g_{C_L, p}(x, z) \equiv C_L(x, z) \mod p$ when all entries of $x$ and $z$ are restricted to $\{\, 0, 1 \,\}$ values. It is straightforward to see that the corresponding polynomial of each gate computes the output of that gate over $\mathbb{Z}_p$. These polynomials generalize Boolean circuits to take inputs over $\mathbb{Z}_p$.

With this construction, we can prove the following lemma for $\NP$-complete languages.

\begin{lemma} 
\label{lemma:5}
\textbf{Generalized Certificate Counting Polynomials.} \\
For any $\NP$-complete language $L$, there is a polynomial $f_{L, p}: \mathbb{Z}_p^n \to \mathbb{Z}_p$ with coefficients in $\mathbb{Z}_p$ and degree bounded by some polynomial in $|x| = n$, computing the number of accepting certificates over $\mathbb{Z}_p$ of the fact that $x \in L$ for $x \in \{\, 0, 1 \,\}^n$.
\end{lemma}
\begin{proof} 
Given any circuit $C_L$ of size $s$ and depth $d$, the corresponding polynomial has degree at most $O \left( s^d \right)$. This can be seen due to the following recurrence relation:
\begin{equation*}
\text{deg}_{\max}(d) \leq s \cdot \text{deg}_{\max}(d - 1),
\end{equation*}
where $\text{deg}_{\max}(d)$ is the largest degree of any corresponding polynomial of a circuit of size $s$ and depth $d$. If $s$ is a polynomial in $n$, and $d$ is constant, then the degree $\text{deg}(g_{C_L, p})$ of $g_{C_L, p}$ is bounded by a polynomial in $n$.

Due to the theorems of \cite{Cook1971} and \cite{Levin1973}, for any language $L$ in $\NP$, there is a polynomial-sized circuit $C_L$ taking an input $x \in \{\, 0, 1 \,\}^n$ and a potential certificate $z \in \{\, 0, 1 \,\}^{n^c}$ and outputting $1$ if $z$ is an accepting certificate for $x \in L$ and $0$ otherwise. We will prove that $C_L$ has a constant depth, implying that the corresponding polynomial $g_{C_L, p}$ has a degree bounded by a polynomial in $n$.

There is a constant depth polynomial-sized circuit computing the bits of reduction from any language in $\NP$ to $\SAT$ \citep{Agrawal2001}. Moreover, the description of this circuit is computable in poly-logarithmic time. Now that we have generated the bits of the $\SAT$ formula, we can use the constant degree verifier for the $\SAT$ formula. The depth of this new circuit, $C_L$, is the depth of the reduction circuit plus the depth of the $\SAT$ verifier. This complete circuit is in $\AC^0$, and its description is computable in polynomial-time.

Consider the following polynomial:
\begin{equation} 
\label{eq:1}
f_{L, p}(x) = \sum_{z \in \{\, 0, 1 \,\}^{n^c}} g_{C_L, p} (x, z).
\end{equation}
Here, $f_{L, p}(x)$ can be seen as a polynomial computing the number of accepting certificates over $\mathbb{Z}_p$ for the fact that $x \in L$, provided that $x$ has entries only in $\{\, 0, 1 \,\}$. Note that $\text{deg}(f_{L, p})$ is bounded by a polynomial in $n$ due to the fact that $\text{deg}(f_{L, p}) = \text{deg}(g_{C_L, p})$.
\end{proof}

We further generalize the functions $f_{L, p}$ and $g_{C_L, p}$ to $f^\prime_{L, p}: \mathbb{Z}_p^{n + 2 n^c} \to \mathbb{Z}_p$ and $g^\prime_{C_L, p}: \mathbb{Z}_p^{n + 3 n^c} \to \mathbb{Z}_p$ respectively, by introducing new variables $a \in \mathbb{Z}_p^{n^c}$ and $b \in \mathbb{Z}_p^{n^c}$ as follows:
\begin{equation} 
\label{eq:2}
g^\prime_{C_L, p}(x, z, a, b) = g_{C_L, p}(x, az + b),
\end{equation}
and using Equation \eqref{eq:2},
\begin{equation} 
\label{eq:3}
f^\prime_{L, p}(x, a, b) = \sum_{z \in \{\, 0, 1 \,\}^{n^c}} g^\prime_{C_L, p}(x, z, a, b).
\end{equation}
Equation \eqref{eq:1} is a special case of Equation \eqref{eq:3}: putting $a = (1)_{i = 1}^{n^c}$ and $b = (0)_{i = 1}^{n^c}$ in Equation \eqref{eq:3}, we get Equation \eqref{eq:1} which is a certificate counting polynomial over $\mathbb{Z}_p$ of the fact that $x \in L$. The motivation for this change is the following lemma.
\begin{lemma}
\label{lemma:6}
\textbf{The Self-Reduction Property of $f^\prime_{L, p}$.} 
\begin{equation} 
\label{eq:4}
\begin{split}
& \sum_{\left( (z_j)_{j = 1}^{i - 1}, (z_j)_{j = i + 1}^{n^c} \right) \in \{\, 0, 1 \,\}^{n^c - 1}} g^\prime_{C_L, p} \left( x, ((z_j)_{j = 1}^{i - 1}, r, (z_j)_{j = i + 1}^{n^c}, a, b \right) \\
& = 2^{-1} \cdot f^\prime_{L, p} \left( x, (a_j)_{j = 1}^{i - 1}, 0, (a_j)_{j = i + 1}^{n^c}, (b_j)_{j = 1}^{i - 1}, a_ir + b_i, (b_j)_{j = i + 1}^{n^c} \right), \\
& \forall i \in [n^c].
\end{split}
\end{equation}
\end{lemma}
\begin{proof}
Using Equations \eqref{eq:2} and \eqref{eq:3}, we get
\begin{equation} 
\label{eq:5}
\begin{split}
& f^\prime_{L, p} \left( x, (a_j)_{j = 1}^{i - 1}, 0, (a_j)_{j = i + 1}^{n^c}, (b_j)_{j = 1}^{i - 1}, a_ir + b_i, (b_j)_{j = i + 1}^{n^c} \right) \\
& = \sum_{z \in \{\, 0, 1 \,\}^{n^c}} g^\prime_{C_L, p} \left( x, z, (a_j)_{j = 1}^{i - 1}, 0, (a_j)_{j = i + 1}^{n^c}, (b_j)_{j = 1}^{i - 1}, a_ir + b_i, (b_j)_{j = i + 1}^{n^c} \right) \\
& = \sum_{z \in \{\, 0, 1 \,\}^{n^c}} g_{C_L, p} \left( x, \left( (a_j)_{j = 1}^{i - 1}, 0, (a_j)_{j = i + 1}^{n^c} \right) z + \left( (b_j)_{j = 1}^{i - 1}, a_ir + b_i, (b_j)_{j = i + 1}^{n^c} \right) \right) \\
& = \sum_{\left( (z_j)_{j = 1}^{i - 1}, (z_j)_{j = i + 1}^{n^c} \right) \in \{\, 0, 1 \,\}^{n^c - 1}} g_{C_L, p} \left( x, \left( (a_j)_{j = 1}^{i - 1}, 0, (a_j)_{j = i + 1}^{n^c} \right) \left( (z_j)_{j = 1}^{i - 1}, 0, (z_j)_{j = i + 1}^{n^c} \right) \right. \\
& \left. + \left( (b_j)_{j = 1}^{i - 1}, a_ir + b_i, (b_j)_{j = i + 1}^{n^c} \right) \right) \\
& + \sum_{\left( (z_j)_{j = 1}^{i - 1}, (z_j)_{j = i + 1}^{n^c} \right) \in \{\, 0, 1 \,\}^{n^c - 1}} g_{C_L, p} \left( x, \left( (a_j)_{j = 1}^{i - 1}, 0, (a_j)_{j = i + 1}^{n^c} \right) \left( (z_j)_{j = 1}^{i - 1}, 1, (z_j)_{j = i + 1}^{n^c} \right) \right. \\
& \left. + \left( (b_j)_{j = 1}^{i - 1}, a_ir + b_i, (b_j)_{j = i + 1}^{n^c} \right) \right) \\
& = \sum_{\left( (z_j)_{j = 1}^{i - 1}, (z_j)_{j = i + 1}^{n^c} \right) \in \{\, 0, 1 \,\}^{n^c - 1}} g_{C_L, p} \left( x, a \left( (z_j)_{j = 1}^{i - 1}, r, (z_j)_{j = i + 1}^{n^c} \right) + b \right) \\
& + \sum_{\left( (z_j)_{j = 1}^{i - 1}, (z_j)_{j = i + 1}^{n^c} \right) \in \{\, 0, 1 \,\}^{n^c - 1}} g_{C_L, p} \left( x, a \left( (z_j)_{j = 1}^{i - 1}, r, (z_j)_{j = i + 1}^{n^c} \right) + b \right) \\
& = \sum_{\left( (z_j)_{j = 1}^{i - 1}, (z_j)_{j = i + 1}^{n^c} \right) \in \{\, 0, 1 \,\}^{n^c - 1}} g^\prime_{C_L, p} \left( x, \left( (z_j)_{j = 1}^{i - 1}, r, (z_j)_{j = i + 1}^{n^c} \right), a, b \right) \\
& + \sum_{\left( (z_j)_{j = 1}^{i - 1}, (z_j)_{j = i + 1}^{n^c} \right) \in \{\, 0, 1 \,\}^{n^c - 1}} g^\prime_{C_L, p} \left( x, \left( (z_j)_{j = 1}^{i - 1}, r, (z_j)_{j = i + 1}^{n^c} \right), a, b \right) \\
& = 2 \sum_{\left( (z_j)_{j = 1}^{i - 1}, (z_j)_{j = i + 1}^{n^c} \right) \in \{\, 0, 1 \,\}^{n^c - 1}} g^\prime_{C_L, p} \left( x, \left( (z_j)_{j = 1}^{i - 1}, r, (z_j)_{j = i + 1}^{n^c} \right), a, b \right).
\end{split}
\end{equation}
Multiplying Equation \eqref{eq:5} by $2^{-1}$, we get Equation \eqref{eq:4}.
\end{proof}

Similarly, we can build the polynomial $g_{C_L}: \mathbb{Z}^{n + n^c} \to \mathbb{Z}$ such that $g_{C_L}(x, z) = C_L(x, z)$ when all entries of $x$ and $z$ are restricted to $\{\, 0, 1 \,\}$ values. Similar to Equation \eqref{eq:1}, we can also build the certificate counting polynomial $f_L: \mathbb{Z}^n \to \mathbb{Z}$ such that $f_L(x)$ counts the number of certificates over $\mathbb{Z}$ for the fact that $x \in L$ when all entries of $x$ are restricted to $\{\, 0, 1 \,\}$ values. Similar to Equations \eqref{eq:2} and \eqref{eq:3}, we can also build the generalized functions $f^\prime_L: \mathbb{Z}^{n + 2 n^c} \to \mathbb{Z}$ and $g^\prime_L: \mathbb{Z}^{n + 3 n^c} \to \mathbb{Z}$, respectively.

With this property, we can ``simulate'' the sumcheck protocol (Section \ref{section:2.6}) and verify the answers. Our idea is to attempt to compute $f^\prime_{L, p}(x)$ for sufficiently many primes $p$, verify those answers, and reconstruct $f^\prime_L \left( x, (1)_{i = 1}^{n^c}, (0)_{i = 1}^{n^c} \right)$ over $\mathbb{Z}$, where $f^\prime_L$ is constructed from the verifier circuit of an $\NP$-complete problem.

\section{The Oracle Sumcheck Protocol over $\mathbb{Z}_p$}
\label{section:5}

In Section \ref{section:4}, we computed a ``generalized'' certificate counting polynomial $f^\prime_{L, p}$ for an $\NP$-complete language $L$ having a polynomial-time computable verification circuit of polynomial-size and constant depth. Suppose we are given an oracle $O$ for computing $f^\prime_{L, p}$ that is correct on a $1/n^\alpha$-fraction of instances of input size $n$. Our idea is to use the \textit{Oracle Sumcheck Protocol (OSP)} (Algorithm \ref{alg:1}) that implements the ideas introduced in Section \ref{section:2.6}. It simulates an interactive proof between us and the oracle $O$, assisted by the STV list-decoder (Lemma \ref{lemma:2}) to help it amplify from $1/n^\alpha$-correctness. Now, the STV list-decoder will provide us with $\poly(n)$ many machines that compute different polynomials, each with a reasonably high agreement with $O$, and the task that remains is to identify which one of these machines computes $f^\prime_{L, p}$ or find out if none of them do. We use the OSP algorithm to verify whether the answer given by a machine $M^O$ is correct for $f^\prime_{L, p}$.

\begin{algorithm}
\caption{$\text{OSP}\left( M^O, x, a, b, p \right)$}
\label{alg:1}
\Comment{The Oracle Sumcheck Protocol \hspace{10.5cm}} \\
\Comment{All computations are done over $\mathbb{Z}_p$} \\
\Comment{The oracle $O$ computes the function $f^\prime_{L, p}$ that is correct on a $1/n^\alpha$-fraction of instances of input size $n$} \\
\Comment{Inputs: $M^O$ is a randomized oracle machine that may compute $f^\prime_{L, p}$, $x \in \mathbb{Z}_p^n$, $a \in \mathbb{Z}_p^{n^c}$,  $b \in \mathbb{Z}_p^{n^c}$}, $p$ is a prime \\
\Comment{Output: \texttt{ACCEPT} if $M^O(x, a, b) = f^\prime_{L, p}(x, a, b)$, \texttt{REJECT} otherwise}
\begin{algorithmic}
\State $i \gets 0$
\While{$i < n^c$} 
    \State $s_i \gets M^O(x, a, b, p)$
    \State $i \gets i + 1$
    \State $c \gets a_i$
    \State $d \gets b_i$
    \For{each $r \in \mathbb{Z}_p$ \do} \Comment{Applying Lemma \ref{lemma:6} to get a new value of $a$ and $b$}
          \State $z_i \gets r$ \Comment{$(z_j)_{j = 0}^{i - 1}$ are already fixed in Lemma \ref{lemma:6}}
          \State $a_i \gets 0$
          \State $b_i \gets cr + d$
          \State $s_{ir} \gets 2^{-1} \cdot M^O(x, a, b, p)$ \Comment{Computing the LHS of Equation \eqref{eq:4}}
     \EndFor
     \State using polynomial interpolation, compute $h_i : \mathbb{Z}_p \to \mathbb{Z}_p$ such that $h_i(r) = s_{ir}, \quad \forall r \in \mathbb{Z}_p$.
     \If{the degree of $h_i$ is greater than that of $z_i$ in $g^\prime_{C_L, p}$}
        \State \Return \texttt{REJECT}
     \EndIf
     \If{$h_i(0) + h_i(1) \neq s_{i - 1}$}
        \State \Return \texttt{REJECT}
     \EndIf
     \State pick a random $r_i \in \mathbb{Z}_p$
     \State $z_i \gets r_i$
     \State $a_i \gets 0$
     \State $b_i \gets cr_i + d$
\EndWhile
\If{$h_i(r_i) = g^\prime_{C_L, p}(x, a, b)$} \Comment{Use the circuit $C_L$ to compute $g^\prime_{C_L, p}$}
        \State \Return \texttt{ACCEPT}
\Else \State \Return \texttt{REJECT}
     \EndIf
\end{algorithmic}
\end{algorithm}

Notice that OSP takes $\poly(p, n)$-time to run. If all primes $p$ we consider are bounded by a polynomial in $n$, then the verification process requires $\poly(n)$-time, including calls to $O$. Using OSP, we can recover the polynomial $f^\prime_{L, p}$ in polynomial-time with very small probability of error using the following lemma.

\begin{lemma} 
\label{lemma:7}
If any oracle $O$ correctly computes the polynomial $f^\prime_{L, p}$ of degree $d$ over $\mathbb{Z}_p$ on more than a $1 / n^\alpha$-fraction of instances with $|x| = n$ and $p > 2 d n^{2 \max \{\, \alpha, c \,\}}$, then with an error probability of at most $1 / 2^{q(n)}$ ($q$ is a polynomial), $f^\prime_{L, p}$ can be computed in $\poly(n)$-time, provided that $p$ grows as a polynomial in $n$.
\end{lemma}
\begin{proof}
Suppose $O$ is correct on more than a $1 / n^\alpha$-fraction of instances over $\mathbb{Z}_p$. Due to Lemma \ref{lemma:2}, the STV list-decoder returns $O(n^\alpha)$ randomized oracle machines that err with probability at most $1 / 2^{q_1(n)}$ for some polynomial $q_1$, each computing a polynomial that is ``close'' to $O$. One of these machines computes $f^\prime_{L, p}$\footnote{All of these machines compute polynomials with the same domain and range as $f^\prime_{L, p}$.}. Once we identify this machine, our job is complete, since the machine $M^O_{f^\prime_{L, p}}$ associated with $f^\prime_{L, p}$ computes it with exponentially low error probability.

For each machine $M^O_f$, we use OSP (Algorithm \ref{alg:1}) on any input. We argue that $M^O_{f^\prime_{L, p}}$ passes this protocol with high probability and that any machines computing other polynomials fail with high probability. This proof follows the same line of reasoning as that in the original paper introducing the sumcheck protocol by \cite{Lund1992}.

We first attempt to show that OSP passes with high probability for $M^O_{f^\prime_{L, p}}$. Notice that if $M^O_{f^\prime_{L, p}}$ computed $f^\prime_{L, p}$ with probability $1$, it would pass the protocol with probability $1$. If $M^O_{f^\prime_{L, p}}$ makes no errors in computing the queries given to it, then the protocol passes. Since the probability that any query is incorrect is at most $1 / 2^{q_1(n)}$ for some polynomial $q_1$, and we have $\poly(p, n) = \poly(n)$ queries, we can union-bound the probability that at least one error occurs.
\begin{equation}
\label{eq:6}
\begin{split}
& \mathcal{P} \left[ M^O_{f^\prime_{L, p}} \text{ makes a mistake on at least one query} \right] \\
& = \mathcal{P} \left[ \cup_{j = 1}^{\poly(n)} \left( M^O_{f^\prime_{L, p}} \text{ makes a mistake on query } j \right) \right] \\
& \leq \sum_{j = 1}^{\poly(n)} \frac{1}{2^{q_1(n)}} \\
& = \frac{\poly(n)}{2^{q_1(n)}} \\
& \leq \frac{1}{2^{q_2(n)}},
\end{split}
\end{equation}
for some polynomial $q_2(n)$.

For any $M^O_f$ at all, if $s_0 \neq f^\prime_{L, p}(x, a, b)$, then with high probability, OSP rejects. We argue that for an OSP that requires $k$ variable settings, the probability of passing the protocol is bounded from above by 
\begin{equation} 
\label{eq:7}
\frac{dk}{|\mathbb{Z}_p|} = \frac{dk}{p}. 
\end{equation}
We will prove this proposition by induction on $k$. Let $P(k)$ be the induction hypothesis as given above in Equation \eqref{eq:7}. \\
\textit{Basis Step:} For $k = 0$, where $s_0 \neq f^\prime_{L, p}(x, a, b)$, we can compute $f^\prime_{L, p}(x, a, b)$ and reject with the probability of passing to be at most $0$, implying that $P(0)$ is correct. \\
\textit{Induction Step:} Using strong induction, assuming that $P(k)$ is true for all $k \in \{\, 0 \,\} \cup [K - 1]$, we will prove $P(K)$. Suppose $s_0 \neq f^\prime_{L, p}(x, a, b)$ and $K$ variables are yet to be set. We will have constructed a polynomial $h_1$ of degree at most $d$ with the help of $M^O$ using polynomial interpolation\footnote{If we do not check the degree of $h_1$, it can ``fool'' us.}. Note that if $h_1$ does not coincide with the corresponding sum of $g^\prime_{C_L, p}$ (Equation \eqref{eq:4}), then when we choose an element $r_1$ from $\mathbb{Z}_p$, the probability that $h_1(r_1)$ is equal to the sum of $g^\prime_{C_L, p}$ (Equation \eqref{eq:4}), with the appropriate parameter set to $r_1$ is at most $d / p$, using the Schwartz-Zippel Lemma (\ref{lemma:1}). Hence, we obtain the following:
\begin{equation}
\label{eq:8}
\begin{split}
& \mathcal{P}[\text{OSP passes with $K$ settings remaining}] \leq \mathcal{P} \left[ \text{$h_1(r_1)$ coincides with the sum of $g^\prime_{C_L, p}$} \right] \\
& + \mathcal{P} \left[\text{$h_1(r_1)$ does not coincide with the sum of $g^\prime_{C_L, p}$ but passes with $K - 1$ steps remaining} \right] \\
& \leq \frac{d}{p} \cdot 1 + \frac{d(K - 1)}{p} \cdot 1 \\
& = \frac{dK}{p} \\
& \leq \frac{1}{n^c},
\end{split}
\end{equation}
since $p > 2 d n^{2c}$ and $K = n^c$, implying that $P(K)$ is true.

We can repeat the protocol $\poly(n)$ many times and approve this machine if it passes every time. From equation \eqref{eq:6}, due to the union-bound, we can still keep the probability of $M^O_{f^\prime_{L, p}}$ failing to be exponentially low,
\begin{equation*}
\frac{\poly(n)}{2^{q_2(n)}} \leq \frac{1}{2^{q(n)}},
\end{equation*}
for some polynomial $q$. From Equation \eqref{eq:8}, the probability of some $M \neq M^O_{f^\prime_{L, p}}$ passing all times is bounded above by 
\begin{equation*}
\left( \frac{1}{n^c} \right)^{\poly(n)} \leq \frac{1}{2^{q_4(n)}},
\end{equation*}
for some polynomial $q_4$. Due to the polynomial union-bound, the probability that even one $M \neq M^O_{f^\prime_{L, p}}$ that the STV list-decoder gives us passes is
\begin{equation*}
\frac{n^\alpha - 1}{2^{q_4(n)}} \leq \frac{1}{2^{q(n)}},
\end{equation*}
for some polynomial $q$.

Once we find $M^O_{f^\prime_{L, p}}$, we can ask it for the original computation we needed. We can also verify it using OSP. We could have also done this as our original verification mechanism to find $M^O_{f^\prime_{L, p}}$. In all cases, with the provided sufficient correctness from $O$, we can compute $f^\prime_{L, p}$ with exponentially low error probability.
\end{proof}

\section{Reconstructing the Certificate Counting Polynomials over $\mathbb{Z}$} 
\label{section:6}

In this section, we will reconstruct the certificate counting polynomial, $f^\prime_L \left( x, (1)_{i = 1}^{n^c}, (0)_{i = 1}^{n^c} \right)$, over $\mathbb{Z}$ using the Chinese remainder theorem (Lemma \ref{lemma:3}). The number of possible certificates can be up to $2^{n^c}$. Therefore, if we use the Chinese remainder theorem on more than $n^c$ distinct primes $(p_i)_{i = 1}^{n^c}$, we will get the correct answers because $\prod_{i = 1}^{n^c} p_i > 2^{n^c}$.

Now, for each $\beta > 0$, we define the number-theoretic function
\begin{equation*}
f^{\prime \prime}_{L, \beta} : \cup_{p \in \left( n^\beta, n^\beta + \frac{n^\beta}{n + 2 n^c} \right)_p} \mathbb{Z}_p^n \times \mathbb{Z}_p^{n^c} \times \mathbb{Z}_p^{n^c} \times \{\, p \,\} \to \cup_{p \in \left( n^\beta, n^\beta + \frac{n^\beta}{n + 2 n^c} \right)_p} \mathbb{Z}_p \times \{\, p \,\},
\end{equation*}
to be computed by any potential oracle $O$ as
\begin{equation*}
f^{\prime \prime}_{L, \beta}(x, a, b, p) = \left( f^\prime_{L, p}(x, a, b), p \right),
\end{equation*}
where, $p \in \left( n^\beta, n^\beta + \displaystyle \frac{n^\beta}{n + 2 n^c} \right)_p$, $x \in \mathbb{Z}_p^n$, $a \in \mathbb{Z}_p^{n^c}$, and $b \in \mathbb{Z}_p^{n^c}$. We will now prove the following lemma that enables us to reconstruct the certificate counting polynomial, $f^\prime_L \left( x, (1)_{i = 1}^{n^c}, (0)_{i = 1}^{n^c} \right)$, over $\mathbb{Z}$.

\begin{lemma}
\label{lemma:8}
\textbf{Reconstructing the Certificate Counting Polynomials over $\mathbb{Z}$.} \\
For each $\alpha > 0$, there is a $\beta > 0$ such that if $f^{\prime \prime}_{L, \beta}$ is computable by an oracle $O$ on a $1 / n^\alpha$-fraction of instances, then we can reconstruct the certificate counting polynomial, $f^\prime_L \left( x, (1)_{i = 1}^{n^c}, (0)_{i = 1}^{n^c} \right)$, over $\mathbb{Z}$ with high probability.
\end{lemma}
\begin{proof}
Let
\begin{equation}
\label{eq:9}
\beta = 3 + 10^6 (\alpha + c).
\end{equation}
Lemma \ref{lemma:4} implies that the largest prime gap in the interval $\left( n^\beta, n^\beta + \displaystyle \frac{n^\beta}{n + 2 n^c} \right)$ is of the order of
\begin{equation}
\label{eq:10}
\begin{split}
O \left( \left( n^\beta \left( 1 + \frac{1}{n + 2 n^c} \right) \right)^{0.525} \right) & = O \left( \left( 2 n^\beta \right)^{0.525} \right) \\
& = O \left( n^{0.525 \beta} \right).
\end{split}
\end{equation}
Using equations \eqref{eq:9} and \eqref{eq:10}, the number of primes in the interval $\left( n^\beta, n^\beta + \displaystyle \frac{n^\beta}{n + 2 n^c} \right)$ is given by
\begin{equation}
\label{eq:11}
\begin{split}
\pi \left( n^\beta, n^\beta + \frac{n^\beta}{n + 2 n^c} \right) & = \Omega \left( \frac{\displaystyle \frac{n^\beta}{n + 2 n^c}}{n^{0.525 \beta}} \right) \\
& = \Omega \left( \frac{n^{0.475 \beta}}{n + 2 n^c} \right) \\
& = \Omega \left( n^{0.474 \beta} \right).
\end{split}
\end{equation}

Now that we have proved that there are many primes, we must prove that a sufficient fraction of these primes have sufficient correctness. By sufficient correctness of a prime $p$, we mean that $O$ must compute $f^\prime_{L, p}$ correctly on more than a $1 / n^{2 \alpha}$-fraction of instances. There are $p^{n + 2 n^c}$ instances of $f^\prime_{L, p}$ satisfying the following inequality\footnote{$\left( 1 + \frac{1}{m} \right)^m$ is increasing and $\lim_{m \to \infty} \left( 1 + \frac{1}{m} \right)^m = e$.}:
\begin{equation}
\label{eq:12}
n^{\left( n + 2 n^c \right) \beta} < p^{n + 2 n^c} < n^{\left( n + 2 n^c \right) \beta} \left( 1 + \frac{1}{n + 2 n^c} \right)^{n + 2 n^c} < e n^{\left( n + 2 n^c \right) \beta}.
\end{equation}

Informally, the number of instances for each prime is balanced up to a constant factor. The following observation holds within this range for primes $p_1$ and $p_2$:
\begin{equation*}
p_1^{n + 2 n^c} < e p_2^{n + 2 n^c}.
\end{equation*}
From the above observations, a lower bound on the number of correct instances we must have over all primes must be
\begin{equation}
\label{eq:13}
\frac{n^{\left( n + 2 n^c \right) \beta} \pi \left( n^\beta, n^\beta + \displaystyle \frac{n^\beta}{n + 2 n^c} \right)}{n^\alpha}.
\end{equation}

Let the random variables $X: \left( n^\beta, n^\beta + \frac{n^\beta}{n + 2 n^c} \right)_p \to [0, 1]$ and $Y: \left( n^\beta, n^\beta + \frac{n^\beta}{n + 2 n^c} \right)_p \to [0, 1]$ be defined as
\begin{equation*}
X(p) = \text{ the fraction of correct answers for instances of } p,
\end{equation*}
and
\begin{equation*}
Y(p) = \text{ the fraction of incorrect answers for instances of } p,
\end{equation*}
respectively. Using Equations \eqref{eq:12} and \eqref{eq:13}, we have
\begin{equation}
\label{eq:14}
\begin{split}
\mathcal{E}[X] & = \sum_{p \in \left( n^\beta, n^\beta + \frac{n^\beta}{n + 2 n^c} \right)_p} \frac{X(p)}{\pi \left( n^\beta, n^\beta + \displaystyle \frac{n^\beta}{n+2n^c} \right)} \\
& \geq \sum_{p \in \left( n^\beta, n^\beta + \frac{n^\beta}{n + 2 n^c} \right)_p} \frac{p^{n + 2 n^c} X(p)}{e n^{\left( n + 2 n^c \right) \beta} \pi \left( n^\beta, n^\beta + \displaystyle \frac{n^\beta}{n+2n^c} \right)} \\
& \geq \frac{n^{\left( n + 2 n^c \right) \beta} \pi \left( n^\beta, n^\beta + \displaystyle \frac{n^\beta}{n+2n^c} \right)}{e n^\alpha \cdot n^{\left( n + 2 n^c \right) \beta} \pi \left( n^\beta, n^\beta + \displaystyle \frac{n^\beta}{n+2n^c} \right)} \\
& = \frac{1}{e n^\alpha}.
\end{split}
\end{equation}
Using Equation \eqref{eq:14} and linearity of expectation, we have
\begin{equation}
\label{eq:15}
\begin{split}
\mathcal{E}[Y] & = 1 - \mathcal{E}[X] \\
& \leq 1 - \frac{1}{e n^\alpha}.
\end{split}
\end{equation}
Now, suppose we pick a random prime uniformly from this restricted range of primes. Using Equation \eqref{eq:15} and Markov's inequality \citep{Mitzenmacher2005}, we have the following\footnote{Since $\displaystyle \frac{1}{1-x} < 1+2x$ for sufficiently small positive $x$, due to the Taylor series.}:
\begin{equation}
\label{eq:16}
\begin{split}
\mathcal{P} \left[ X(p) \leq \displaystyle \frac{1}{n^{2\alpha}} \right] & = \mathcal{P} \left[ Y(p) \geq 1 - \displaystyle \frac{1}{n^{2\alpha}} \right] \\
& \leq \frac{\mathcal{E}[Y]}{1 - \displaystyle \frac{1}{n^{2\alpha}}} \\
& \leq \frac{1 - \displaystyle \frac{1}{en^\alpha}}{1 - \displaystyle \frac{1}{n^{2\alpha}}} \\
& \leq \left( 1 - \frac{1}{en^\alpha} \right) \left( 1 + \frac{2}{n^{2\alpha}} \right) \\
& = 1 - \Omega \left( \frac{1}{n^\alpha} \right).
\end{split}
\end{equation}
Hence, from Equation \eqref{eq:16}, the probability that we have more than $1/n^{2\alpha}$-fraction of correctness for $p$ is given by
\begin{equation}
\mathcal{P} \left[ X(p) > \frac{1}{n^{2\alpha}} \right] = \Omega \left( \frac{1}{n^\alpha} \right).
\end{equation}
\end{proof}
From Equations \eqref{eq:9}, \eqref{eq:11}, and \eqref{eq:16}, the number of primes for which we have sufficient correctness is 
\begin{equation*}
\Omega \left( \frac{n^{0.474 \beta}}{n^\alpha} \right) = \Omega \left( n^{0.473 \beta} \right).
\end{equation*}

Due to Lemma \ref{lemma:7}, each prime with sufficient correctness correctly computes the answer to $f^\prime_{L, p}$ with exponentially low error probability, and each prime with insufficient correctness either gives us the answer by accident or rejects all machines with exponentially low error probability. We can union-bound all the errors to exponentially low probability since the union would cause a $\poly(n)$-fold increase of a decreasing exponential function.

Using the Chinese remainder theorem (Lemma \ref{lemma:3}), from all correct answers to $f^\prime_{L, p}$, we can reconstruct $f^\prime_L \left( x, (1)_{i = 1}^{n^c}, (0)_{i = 1}^{n^c} \right)$ over $\mathbb{Z}$. This occurs in $\poly(n)$ time and queries to $O$.

\section{Main Results} 
\label{section:7}
We are now ready to state the main theorem of this paper. Unless otherwise specified, $n$ is the size of the input string $x$ ($n = |x|$) to the language membership problem of $x \in L$.
\begin{theorem}
\label{theorem:1}
Given any $\NP$-complete language $L$, for any $\alpha > 0$, we have a function $f^{\prime\prime}_{L, \beta}$ such that given an oracle $O$ that computes $f^{\prime\prime}_{L, \beta}$ correctly on a $1/n^{\alpha}$-fraction of instances, for any polynomial $q:\mathbb{N}\to\mathbb{N}$, we have a probabilistic algorithm that decides whether $x \in L$ with error probability at most $1/2^{q(n)}$ in $\poly(n)$-time and $\poly(n)$-queries to $O$.
     
Moreover, we have a polynomial-time proof system in which a verifier $V$ can verify the validity of a claim of the form $s = f^{\prime\prime}_{L, \beta}(x, a, b, p)$ in $\poly(n)$-time using $\poly(n)$-queries to the prover, $P$, of the form $f^{\prime\prime}_{L, \beta}(y, a^\prime, b^\prime, p)$ for the same prime $p$. This protocol can be modified to work when the prover is only correct on a $1/n^{2\alpha}$-fraction of instances over the instances pertaining to the prime $p$.
\end{theorem}
\begin{proof}
 We construct the polynomials $f^{\prime}_{L, p}$ as discussed in Section \ref{section:4}. Given the oracle $O$, we make the necessary queries ($f^{\prime\prime}_{L, \beta}(x, (1)_{i = 1}^{n^c}, (0)_{i = 1}^{n^c}, p)$) for each prime $p$ in the range, and then certify these results as shown in Lemma \ref{lemma:7}. Due to Lemmas \ref{lemma:7} and \ref{lemma:8}, with exponentially small error probability, sufficiently many primes compute correct answers, are all certified, and no wrong answers are certified. This fact simply follows that the error probability of each of these events is exponentially small, and given we have polynomially many events, the union bound still gives us an exponentially small probability of error. Using the Chinese remainder theorem, as stated in Lemma \ref{lemma:3}, we can use all the certified answers to compute the number of certificates of the fact that $x \in L$, $f^{\prime}_L(x, (1)_{i = 1}^{n^c}, (0)_{i = 1}^{n^c})$, over $\mathbb{Z}$.
    
The interactive proof is due to Lemma \ref{lemma:7}. Note that if the prover is always correct or is capable of being always correct, no error correction is required at the verifier's end. However, if the prover is correct on only a $1/n^{2\alpha}$-fraction (given our choice of $\beta = 3 + 10^6(\alpha + c)$ from Equation \eqref{eq:9}) of instances over the prime $p$, we can use the STV list decoder for the proof, on the verifier's end.
\end{proof}
We have the following immediate corollaries of this theorem.
\begin{corollary}
\label{corollary:1}
If $\NP \not\subset \BPP$, then for all $\alpha > 0$, for each $\NP$-complete language $L$, we have a polynomial-time provable function $f^{\prime\prime}_{L, \beta}$ such that no polynomial-time randomized algorithm with error probability less than $1/3$ on correct instances can compute $f^{\prime\prime}_{L, \beta}$ correctly on more than a $1/n^{\alpha}$-fraction of instances.
\end{corollary}
\begin{proof}
For an oracle machine made by the STV list decoder, query the randomized algorithm $\poly(n)$ many times on the same input to have exponentially low error. We will take the value that makes up the majority of the answers if there is one. If there is no such value, we take the answer as $0$. Due to the union-bound over $\poly(n)$ ``queries'' with exponentially small errors on correct instances, the probability that there is even a single mismatch between the answers on correct instances and the values we take down is exponentially small.

If such a randomized algorithm existed, we would have a polynomial-time randomized algorithm for $L$ with exponentially small error bounds on both sides. Due to the $\NP$-completeness of $L$, we would have that $\NP \subset \BPP$.
\end{proof}
\begin{corollary}
\label{corollary:2}
If $\NP \not\subset \PPOLY$, then for all $\alpha > 0$, for each $\NP$-complete language $L$, we have a polynomial-time provable function $f^{\prime\prime}_{L, \beta}$ such that, for all $k > 0$, and all circuit families $\{\, C_n \,\}_{n \in \mathbb{N}}$ computing values from the domain $\mathcal{D}$ of $f^{\prime\prime}_{L, \beta}$ to the range of $f^{\prime\prime}_{L, \beta}$, $C_n$ is of size at most $n^k$ (where $n$ = $|x|$), for sufficiently large $n$, we have
\begin{equation*}
\mathcal{P}_{x \leftarrow_r \mathcal{D}}\left[ f^{\prime\prime}_{L, \beta}(x) = C_n(x) \right] < \frac{1}{n^{\alpha}}.
\end{equation*}
\end{corollary}
\begin{proof}
Due to ideas similar to the theorem of \cite{Adleman1978}\footnote{By simply hard-coding a successful reduction string that works for all $x$ with $|x|=n$.}, our probabilistic reduction extends to the case of circuits. The remainder of the proof proceeds similar to that of Corollary \ref{corollary:1}.
\end{proof}
We now state the $\P^{\SHARPP}$ versions of these corollaries, which rely on weaker conjectures. The proofs proceed identically to their $\NP$ counterparts.
\begin{corollary}
\label{corollary:3}
If $\P^{\SHARPP} \not\subset \BPP$, then for all $\alpha > 0$, for each $\NP$-complete language $L$ that has parsimonious reductions from every language in $\NP$, we have a polynomial-time provable function $f^{\prime\prime}_{L, \beta}$ such that no polynomial-time randomized algorithm with error probability less than $1/3$ on correct instances can compute $f^{\prime\prime}_{L, \beta}$ correctly on more than a $1/n^{\alpha}$-fraction of instances.
\end{corollary}
\begin{corollary}
\label{corollary:4}
If $\P^{\SHARPP} \not\subset \PPOLY$, then for all $\alpha > 0$, for each $\NP$-complete problem $L$ that has parsimonious reductions from every language in $\NP$, we have a polynomial-time provable function $f^{\prime\prime}_{L, \beta}$ such that, for all $k > 0$, and all circuit families $\{\, C_n \,\}_{n \in \mathbb{N}}$ computing values from the domain $\mathcal{D}$ of $f^{\prime\prime}_{L, \beta}$ to the range of $f^{\prime\prime}_{L, \beta}$, $C_n$ is of size at most $n^k$ (where $n$ = $|x|$), for sufficiently large $n$, we have
\begin{equation*}
\mathcal{P}_{x \leftarrow_r \mathcal{D}} \left[ f^{\prime\prime}_{L, \beta}(x) = C_n(x) \right] < \frac{1}{n^{\alpha}}.
\end{equation*}
\end{corollary}
Now, we will state corollaries depending on certain hypotheses stated in Section \ref{section:3}. Depending on one's faith in these hypotheses, one can see these results as either a very strong rare-case hardness result or a potential weakness of the hypothesis. Agnostically, we state the following corollaries.
\begin{corollary}
\label{corollary:5}
If RETH is true, then there is an $\epsilon > 0$ such that for all $\alpha > 0$, for each $\NP$-complete language $L$, we have a polynomial-time provable function $f^{\prime\prime}_{L, \beta}$ such that any randomized algorithm with error probability less than $1/3$ on correct instances computing $f^{\prime\prime}_{L, \beta}$ correctly on more than a $1/n^{\alpha}$-fraction of instances requires $2^{n^{\epsilon}}$ time.
\end{corollary}
\begin{proof}
Assume for the sake of contradiction that no such $\epsilon > 0$ exists. That is, for every $\epsilon' > 0$, there is a $2^{n^{\epsilon^\prime}}$-time randomized algorithm accomplishing this task. Notice that the reduction from $3\SAT$ to $L$ turns an instance of size $n$ to an instance of size $n^a$. Due to this, we will have a $2^{n^{a\epsilon^\prime}}\textit{poly}(n)$-time algorithm for $3\SAT$ for all $\epsilon^\prime > 0$. When $\epsilon^\prime < 1 / a$, this violates the RETH since $2^{n^{\epsilon' a}} = o(2^{\epsilon_0 n})$ for all $\epsilon_0 > 0$.
\end{proof}
\begin{corollary}
\label{corollary:6}
If RSETH is true, then for every $\epsilon > 0$ and for each $\alpha > 0$, there is a $k \in \mathbb{N}$, such that the $f^{\prime\prime}_{L, \beta}$ derived from $k\SAT$ is not computable in $2^{(1-\epsilon)n}$ time even on a $1/n^{\alpha}$-fraction of instances.
\end{corollary}
\begin{proof}
Assume that there is an $\epsilon > 0$ and an $\alpha > 0$ such that for all $k \in \mathbb{N}$, the $f^{\prime\prime}_{L, \beta}$ derived from $k\SAT$ is computable in $2^{(1-\epsilon)n}$-time on a $1/n^{\alpha}$-fraction of instances. Due to the reduction from $k\SAT$ to $f_{L, \beta}^{\prime \prime}$, we have a $2^{(1-\epsilon)n}\poly(n)$-time randomized algorithm for $k\SAT$ for all $k \in \mathbb{N}$, violating RSETH.
\end{proof}

\section{Other Techniques for Constructing Rare-Case Hard Functions} 
\label{section:8}

This section will overview other techniques for constructing rare-case hard functions. The first method uses \textit{Low-Degree Extensions (LDE)} of functions \citep{Sudan2001}. The second method proves that computing the permanent of random matrices is rare-case hard \citep{Feige1996}.

\subsection{Rare-Case Hard Functions Using Low-Degree Extensions}
\label{section:8.1}

Using the notation of \cite{Sudan2001}, let $P: \{\, 0, 1 \,\}^l \to \{\, 0, 1 \,\}$ be a function. We encode the function $P$ as a $2^l$-bit vector given by
\begin{equation*}
P = \left( P(i) \right)_{i = 0^l}^{1^l}.
\end{equation*}

To compute $\LDE(P)$, we view the vector $P$ as a concatenation of values of a multivariate polynomial evaluated over a subset of points. The parameters of this computation are defined as follows. $\mathbb{F}$ is the field used in the computations. The number of variables in the multivariate polynomial is $m$. $\mathbb{H} \subset \mathbb{F}$ is such that $\mathbb{H}^m$ is the subset of points over which the multivariate polynomial is evaluated. $|\mathbb{H}| - 1$ is the maximum degree of each variable in the multivariate polynomial. We find the unique polynomial $\hat{p} : \mathbb{F}^m \to \mathbb{F}$ using polynomial interpolation over the used points in the vector $P$ and setting the remaining unused points in $\mathbb{H}^m$ having $0$ evaluation. We define $\LDE(P)$ as the evaluation of $\hat{p}$ over all points in $\mathbb{F}^m$:
\begin{equation*}
\LDE(P) = \left( \hat{p}(x) \right)_{x \in \mathbb{F}^m}.
\end{equation*}

Using the settings of the above parameters as given in Lemma 25 of \cite{Sudan2001}, we find that the conditions of the STV list decoder are satisfied (Lemma \ref{lemma:2}). We also observe that we can compute $\LDE(P)$ in linear space if we are given access to $P$, implying that computing $\LDE(P)$ is in $\PSPACE$. Therefore, we can also get an interactive proof protocol for it.

Using the above observations, we get analogous results similar to Section \ref{section:7} that computing the function $\LDE(P)$ for an $\NP$-complete language $P$ is rare-case hard for randomized polynomial-time algorithms and also for polynomial-sized circuit families.

\subsection{Rare-Case Hardness of Computing the Permanent of Random Matrices}
\label{section:8.2}

For an $n \times n$ matrix $A = \left( \left( a_{i, j} \right)_{i \in [n]} \right)_{j \in [n]}$, we define the \textit{permanent} of $A$ as
\begin{equation*}
\PER(A) = \sum_{\sigma \in S_n} \prod_{i \in [n]} a_{i, \sigma(i)},
\end{equation*}
where $S_n$ is the group of permutations over $[n]$. \cite{Feige1996} prove that computing $\PER(A)$ for $n \times n$ matrices over $\mathbb{Z}_p$ is $\left( 13 n^3 / p \right)$-rare-case hard for polynomial-time randomized algorithms under the assumption that $\AM \neq \PH$. They also prove that computing $\PER(A)$ for $n \times n$, $0 / 1$ matrices is $\left( 1 / 2^{n^{1 - \epsilon}} \right)$-rare-case hard for some $\epsilon > 0$ for polynomial-time randomized algorithms under the assumption that $\AM \neq \P^{\#\P}$. Their interactive proof protocol is of constant rounds, unlike our $p$-round interactive proof. This allows $p$ to be exponential in $n$, giving superpolynomially rare-case hardness results. Due to this limitation, our techniques cannot give superpolynomially rare-case hardness results.

\section{Further Directions} 
\label{section:9}

We have managed to show that from large families of languages, one can construct variants of the problem that are as hard in terms of time taken to compute even on a small fraction of instances. Some potential future directions are listed below.

\subsection{Construction of Superpolynomially Rare-Case Hard Functions}
\label{section:8.1}
In this paper, we proved that under assumptions like $\NP \not\subset \PPOLY$, for every $\alpha > 0$, there is a function derived that is hard to compute even on a $1/n^\alpha$-fraction of instances. If one proves this theorem with slight change in quantifier order - ``If $\NP \not\subset \PPOLY$, then there is a function that is hard to compute even on a $1/n^\alpha$-fraction of instances for every $\alpha > 0$, and sufficiently large $n$ (depending on $\alpha$)'' with similar properties as the one we showed, in terms of proof protocols, this would be an important intermediate step in showing something like $\NP \not\subset \PPOLY$ or any other such conjecture implying the existence of one-way functions. Assuming they exist, inverting a one-way function is a special case of a problem that is hard to solve in the rare-case, but there is a fast verification protocol. For a one-way function $f$, this rareness is superpolynomial, and the protocol is simply to provide the inverse - the answer to the inversion problem. We showed this for fixed polynomial rareness and a polynomial-time protocol that requires a polynomial number of answers to verify. Can we get these closer to the one-way function inversion case? The algebraic techniques used here to interpolate might limit the potential for superpolynomial rare-case hardness. Are there techniques that can yield similar or even stronger worst-case to rare-case reductions?

It is known that the existence of one-way functions is equivalent to the existence of pseudorandom generators \citep{Hastad1999}. In the field of ``hardness versus randomness,'' it has also been shown that certain hardness assumptions imply the derandomization of $\BPP$, that is, results such as $\P = \BPP$ and $\BPP \subset \cap_{\epsilon > 0}\TIME(2^{n^{\epsilon}})$ \citep{Nisan1994, Impagliazzo1997}. Can such results be shown with assumptions analogous to $\NP \not\subset \PPOLY$?

\subsection{Derandomizing the Reduction}
\label{section:8.2}
Due to the fact that our reduction is randomized, we could only use conjectures such as $\NP \not\subset \BPP$, $\NP \not\subset \PPOLY$, RETH, and RSETH. Suppose this reduction is derandomized or new rare-case hard problems are constructed with fully deterministic worst-case to rare-case reductions. In that case, one can show similar reductions under the assumption that $\P \neq \NP$ or the standard versions of the exponential time hypotheses.

\subsection{Refuting the Strong Exponential Time Hypothesis}
\label{section:8.3}
Developments in the past decade have cast doubt on the validity of the \textit{Strong Exponential Time Hypothesis (SETH)} \citep{Vyas2019, Williams2014, Williams2024}. To refute even the stronger RSETH, can one find $2^{0.999n}$-time algorithms for the functions $f^{\prime\prime}_{L, \beta}$ we derived from $k\SAT$ for all $k$, that are correct on the required, yet small, fraction of instances? It seems more feasible to find algorithms in cases with algebraic symmetries and when one can afford to be correct on only a vanishing fraction of instances. If not under the derived functions in this paper, can one find functions with worst-case to rare-case reductions from $k\SAT$ that are easier to find algorithms for?

\subsection{Rare-Case Hardness for More Natural Problems}
\label{section:8.4}
We have shown rare-case hardness, which works well with arbitrarily and algebraically defined functions. As is the case with the $\DLP$ \citep{Blum1982} and computing the permanent \citep{Feige1996}, can one show rare-case hardness for more natural problems under some reasonable assumptions?

\bibliographystyle{apalike}

\bibliography{main}

\appendix
\section{An Alternative Proof of a Variant of Lipton's Theorem}
\label{appendix:A}

Our work so far can be extended to give an alternative proof of a variant of a theorem of \cite{Lipton1989}. We stress that Lipton's proof of this theorem was groundbreaking since error correction techniques were still in their primitive stages. \cite{Sudan2001} constructed their breakthrough list decoder many years after Lipton's result.

\begin{theorem}
\label{theorem:2}

If $\P^{\SHARPP} \not\subset \PPOLY$, then for any $\PSPACE$-complete language $L$, for infinitely many input lengths $l$, there is a polynomial-time samplable distribution $\mathcal{D}^\prime$ such that for any polynomial-sized circuit family $\{\, C^\prime_l \,\}_{l \in \mathbb{N}}$ (with $C^\prime_l : \{\, 0, 1 \,\}^l \to \{\, 0, 1 \,\}$),
\[
\Pr_{y^\prime \sim \mathcal{D}^\prime}[C^\prime_l(y^\prime) \neq L(y^\prime)] > \Omega \left( \frac{1}{\log l} \right).
\]

\end{theorem}

\begin{proof}
Using Corollary \ref{corollary:4}, we know that if $\P^{\SHARPP} \not\subset\PPOLY$, then the function $f^{\prime \prime}_{\SAT, \beta}$ derived from $\SAT$ (as defined in Section \ref{section:6}) for the parameter $\alpha$ cannot be computed by polynomial-sized circuit families, even on a $1/n^\alpha$-fraction of instances for sufficiently large $n$ (depending on the size exponent and $\alpha$).

Now, the function
\begin{equation*}
f^{\prime \prime}_{\SAT, \beta} : \cup_{p \in \left( n^\beta, n^\beta + \frac{n^\beta}{n + 2 n^c} \right)_p} \mathbb{Z}_p^n \times \mathbb{Z}_p^{n^c} \times \mathbb{Z}_p^{n^c} \times \{\, p \,\} \to \cup_{p \in \left( n^\beta, n^\beta + \frac{n^\beta}{n + 2 n^c} \right)_p} \mathbb{Z}_p \times \{\, p \,\},
\end{equation*}
can be seen as a function
\begin{equation*}
f^{\prime \prime}_{\SAT, \beta} : \{\, 0, 1 \,\}^{\left( 2 n^c + n + 1 \right) k} \to \{\, 0, 1 \,\}^{2k},
\end{equation*}
where
\begin{equation*}
k = \left \lceil \log \left( n^\beta + \frac{n^\beta}{n + 2 n^c} \right) \right \rceil.
\end{equation*}

Let $L$ be any $\PSPACE$-complete language. Representing $f^{\prime \prime}_{\SAT, \beta}$ as a function from binary strings to binary strings, $f^{\prime \prime}_{\SAT, \beta, i}(x)$ is the $i$'th bit of $f^{\prime \prime}_{\SAT, \beta}(x)$. Since there is an interactive proof for $f^{\prime \prime}_{\SAT, \beta, i}$ (the same interactive proof for $f^{\prime \prime}_{\SAT, \beta}$), the binary language defined by $f^{\prime \prime}_{\SAT, \beta, i}(x)$ (where $x$ is in the language if and only if $f^{\prime \prime}_{\SAT, \beta, i}(x) = 1$) is in $\PSPACE$. Due to this, there is a $\poly(n, k)$-time reduction from $f^{\prime \prime}_{\SAT, \beta, i}$ to $L$. The $L$ instance obtained from $x$ for $f^{\prime \prime}_{\SAT, \beta, i}$ is $y_i$.

We define the function
\begin{equation*}
L_{2k}: \{\, 0, 1 \,\}^{\left( 2 n^c + n + 1 \right) k} \to \{\, 0, 1 \,\}^{2k}
\end{equation*}
as
\begin{equation*}
L(x) = \left( L(y_i) \right)_{i \in [2k]},
\end{equation*}
where each $y_i$ is a string as defined above and $L(y_i)$ acts as an indicator function for $y_i$'s membership in $L$. Doing this for each $i$, we get $L_{2k}(y_i)_{i \in [2k]} = f^{\prime \prime}_{\SAT, \beta}(x)$, the output of $f^{\prime \prime}_{\SAT, \beta, i}$ and the $i$'th bit of $L_{2k}$ being the same. A circuit computing $L_{2k}$ takes $2k$ strings $(y_i)_{i \in [2k]}$ of the same size and returns $\left( L(y_i) \right)_{i \in [2k]}$. 

Let $D$ be the distribution of $y = (y_i)_{i \in [2k]}$ obtained from sampling from the valid instances of $f^{\prime \prime}_{\SAT, \beta}$ uniformly and applying the reduction to $L$ (or $L_{2k}$). Let $D_i$ be the distribution of $y_i$ obtained from sampling $y$ from $D$ and taking only $y_i$. We know from our result in Corollary \ref{corollary:4}, that for any polynomial-sized circuit family $\{\, C_l \,\}_{l \in \mathbb{N}}$,
\begin{equation*}
\Pr_{y \sim D}[C_{2kl}(y) = L_{2k}(y)] < \frac{1}{n^\alpha},
\end{equation*}
where $l = \poly(n, k)$ is the length of each $y_i$ in $y$.

Suppose that we have a polynomial-sized circuit family $\{\, C_l^\prime \,\}_{l \in \mathbb{N}}$. For a given input length $l$, for all $i \in [2k]$, we have
\begin{equation} \label{eq:18}
\Pr_{y_i \sim D_i}[C^\prime_l(y_i) = L(y_i)] = 1 - \epsilon_i.
\end{equation}
There is also a naive polynomial-space algorithm that evaluates all certificates in the sum of $f^{\prime\prime}_{\SAT, \beta}$.

Now, consider the following argument. Given $2k$ copies of $C^\prime_l$, we have a circuit $C_{2kl}$ trying to compute $L_{2k}$. Suppose $X$ is a random variable that, given $y = (y_i)_{i \in [2k]}$, returns the number of bits where $L(y_i) = C^\prime_l(y_i)$ (given $C_{2kl}$ is fed $y$ as input). From the rare-case hardness of $f^{\prime\prime}_{\SAT, \beta}$ (Corollary \ref{corollary:4}), we have
\begin{equation}
\label{eq:19}
\mathbb{E}_{y \sim D}[X] \leq \left(1 - \frac{1}{n^\alpha}\right)(2k - 1) + \frac{2k}{n^\alpha} = 2k - 1 + \frac{1}{n^\alpha}.
\end{equation}

By the linearity of expectation and Equation \eqref{eq:18}, we have
\begin{equation}
\label{eq:20}
\mathbb{E}_{y \sim D}[X] = \sum_{i \in [2k]} \mathbb{E}_{y_i \sim D_i}[X_i(y_i)] = \sum_{i \in [2k]} \Pr_{y_i \sim D_i}[C^\prime_l(y_i) = L(y_i)] = 2k - \sum_{i \in [k]} \epsilon_i
\end{equation}
From Equations \eqref{eq:19} and \eqref{eq:20}, we get
\begin{equation}
\label{eq:21}
\begin{split}
& 2k - 1 + \frac{1}{n^\alpha} \geq 2k - \sum_{i \in [2k]} \epsilon_i \\
& \iff \sum_{i \in [k]} \epsilon_i \geq 1 - \frac{1}{n^\alpha}. 
\end{split}
\end{equation}

We now define the distribution $D^\prime$ as follows: Sample $i$ uniformly from $[2k]$ and then sample from $D_i$. Then, we have
\begin{equation}
\label{eq:22}
\Pr_{y^\prime \sim D^\prime}[C^\prime_l(y^\prime) = L(y^\prime)] = 1 - \frac{\sum_{i \in [2k]} \epsilon_i}{2k}.
\end{equation}
Thus, from Equations \eqref{eq:21} and \eqref{eq:22}, we get
\begin{equation*}
\Pr_{y^\prime \sim D^\prime}[C^\prime_l(y^\prime) \neq L(y^\prime)] = \frac{\sum_{i \in [2k]} \epsilon_i}{2k} \geq \frac{1 - \displaystyle \frac{1}{n^\alpha}}{2k} = \Omega \left( \frac{1}{\log n} \right) = \Omega \left( \frac{1}{\log l} \right).
\end{equation*}
\end{proof}

A similar result holds for $\P^{\SHARPP} \not\subset \BPP$ and polynomial time algorithms. Typically, Yao's XOR Lemma \citep{Goldreich2011} is used to construct harder functions in $\PSPACE$ after obtaining a preliminary hardness amplification.
\end{document}